\RequirePackage{amsmath}
\documentclass[runningheads]{llncs}
\usepackage{multicol}
\usepackage{amsfonts}
\usepackage{amssymb}
\usepackage{graphicx}
\usepackage{epic}
\usepackage{eepic}
\usepackage{epsfig,float}
\usepackage{verbatim}
\usepackage{pdfsync}
\usepackage{gastex}
\usepackage[lowtilde]{url}
\usepackage{enumitem}
\usepackage{url} 

\usepackage{bbm}
\usepackage{xcolor}

\usepackage{xcolor}

\usepackage[utf8]{inputenc}\usepackage[T1]{fontenc}

\renewcommand{\le}{\leqslant}
\renewcommand{\ge}{\geqslant}

\newcommand{\ol}{\overline}
\newcommand{\eps}{\varepsilon}
\newcommand{\emp}{\emptyset}

\newcommand{\Sig}{\Sigma}
\newcommand{\sig}{\sigma}
\newcommand{\noin}{\noindent}

\newcommand{\bi}{\begin{itemize}}
\newcommand{\ei}{\end{itemize}}
\newcommand{\be}{\begin{enumerate}}
\newcommand{\ee}{\end{enumerate}}
\newcommand{\bd}{\begin{description}}
\newcommand{\ed}{\end{description}}
\newcommand{\bq}{\begin{quote}}
\newcommand{\eq}{\end{quote}}

\newcommand{\cD}{{\mathcal D}}
\newcommand{\cE}{{\mathcal E}}

\newcommand{\ra}{{\rightarrow}}

\newcommand{\e}{\emph}
\newcommand{\co}{\colon}

\renewcommand{\phi}{\varphi}

\newcommand{\tr}[1]{\stackrel{#1}{\longrightarrow}}

\title{Most Complex Non-Returning Regular Languages\thanks{This work was supported by the Natural Sciences and Engineering Research Council of Canada 
grant No.~OGP0000871.}
}

\author{Janusz A. Brzozowski\inst{1} \and Sylvie Davies\inst{2}}

\authorrunning{J. A. Brzozowski,  S. Davies}  

\titlerunning{Most Complex Non-Returning Regular Languages}

\institute{David R. Cheriton School of Computer Science, University of Waterloo \\
Waterloo, ON, Canada N2L 3G1\\
{\tt brzozo@uwaterloo.ca}
\and
Department of Pure Mathematics, University of Waterloo \\
Waterloo, ON, Canada N2L 3G1\\
{\tt sldavies@uwaterloo.ca}
}

\begin{document}
\maketitle

\begin{abstract}
A regular language $L$ is non-returning if in the minimal deterministic finite automaton accepting it there are no transitions into the initial state.
Eom, Han and Jir\'askov\'a derived upper bounds on the state complexity of boolean operations and Kleene star, and proved that these bounds are tight using two different binary witnesses.
They derived upper bounds for concatenation and reversal using three different ternary witnesses. These five witnesses use a total of six different transformations. We show that for each $n\ge 4$ there exists a  ternary witness  of state complexity $n$ that meets the bound for reversal and that at least three letters are needed to meet this bound.
Moreover, the restrictions of this witness to  binary alphabets meet the bounds for product, star, and boolean operations. 
We also derive tight upper bounds on the state complexity of binary operations that take arguments with different alphabets. We prove that the maximal syntactic semigroup of a non-returning language has $(n-1)^n$ elements and requires at least $\binom{n}{2}$ generators. We find the maximal state complexities of atoms of non-returning languages. 
Finally, we show that there exists a most complex non-returning language that meets the bounds for all these complexity measures.
\medskip

\noin
\end{abstract}

\section{Introduction} 
Formal definitions are postponed until Section~\ref{sec:preliminaries};
we assume  the reader is familiar with basic properties of regular languages and finite automata as described in~\cite{Per90,Yu97}, for example.

A deterministic finite automaton (DFA) is \emph{non-returning} if there are no transitions into its initial state. 
A regular language is non-returning if its minimal DFA has that property. 
The \emph{state complexity} of a regular language $L$ is the number of states in
the minimal DFA accepting $L$. The state complexity of an \emph{operation} on regular languages is defined as the maximal state complexity of the result of the operation, expressed as a function of the state complexities of the operands.

The state complexity of common operations (union, intersection, difference, symmetric difference, Kleene star, reverse and product/concatenation) was studied by Eom, Han and Jir\'askov\'a~\cite{EHJ16}.
They pointed out that several subclasses of regular languages have the non-returning property; these subclasses include the class of suffix-free languages (suffix codes) and its subclasses  (for example,  bifix-free languages), and finite languages. 

A regular language $L_n(a,b,c)$ of state complexity $n$ is defined for all $n\ge 3$ in Figure~\ref{fig:regular}.
It was shown in~\cite{Brz13} that the sequence $(L_3(a,b,c),\dots, L_n(a,b,c),\dots)$ of these languages meets the upper bounds (for regular languages) on the complexities of all the basic operations on regular languages as follows: 
If $L(b,a)$ is $L(a,b)$ with the roles of $a$ and $b$ interchanged, then$L_m(a,b) \circ L_n(b,a)$ meets the bound $mn$ for all binary boolean operations $\circ$ that depend on both arguments; 
if $m\neq n$, $L_m(a,b) \circ L_n(a,b)$ meets the bound $mn$;
$(L_n(a,b))^*$ meets the bound $2^{n-1} + 2^{n-2}$; 
$(L_n(a,b,c))^R$ meets the bound $2^n$ for reversal;
and $L_m(a,b,c) L_n(a,b,c)$ meets the bound $(m-1) 2^n + 2^{n-1}$ for product.

\begin{figure}[ht]
\unitlength 7.5pt
\begin{center}\begin{picture}(37,5)(0,4)
\gasset{Nh=1.8,Nw=3.5,Nmr=1.25,ELdist=0.4,loopdiam=1.5}
	{\scriptsize
\node(0)(1,7){0}\imark(0)
\node(1)(8,7){1}
\node(2)(15,7){2}
}
\node[Nframe=n](3dots)(22,7){$\dots$}
	{\scriptsize
\node(n-2)(29,7){$n-2$}
	}
	{\scriptsize
\node(n-1)(36,7){$n-1$}\rmark(n-1)
	}
\drawloop(0){$c$}
\drawedge[curvedepth= .8,ELdist=.1](0,1){$a,b$}
\drawedge[curvedepth= .8,ELdist=-1.2](1,0){$b$}
\drawedge(1,2){$a$}
\drawloop(2){$b,c$}
\drawedge(2,3dots){$a$}
\drawedge(3dots,n-2){$a$}
\drawloop(n-2){$b,c$}
\drawedge(n-2,n-1){$a$}
\drawedge[curvedepth= 4.0,ELdist=-1.0](n-1,0){$a,c$}
\drawloop(n-1){$b$}
\drawloop(1){$c$}
\end{picture}\end{center}
\caption{Most complex regular language $L_n(a,b,c)$.}
\label{fig:regular}
\end{figure}
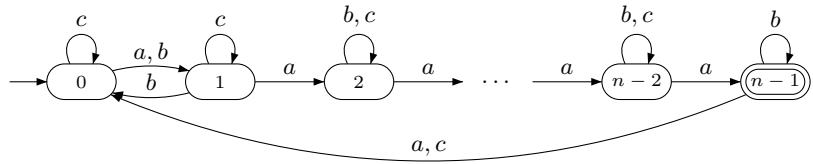

It was proposed in~\cite{Brz13} that the size of the \emph{syntactic semigroup} of a regular language is another worthwhile measure of the complexity of the language.
The syntactic semigroup is isomorphic to  the \emph{transition semigroup} of the minimal DFA of $L$, that is,  the semigroup of transformations of the state set of the DFA induced by non-empty words. 

Another complexity measure suggested in~\cite{Brz13} is the number and state complexities of the atoms of the language. An atom is a certain kind of intersection of complemented and uncomplemented quotients of $L$, where a quotient of $L\subseteq \Sig^*$ by a word $w\in \Sig^*$ is the language $w^{-1}L=\{x\mid wx \in L\}$.

It was shown in~\cite{Brz13} that the languages $L_n(a,b,c)$ not only meet the bounds on the state complexities of operations, but  also have the largest syntactic semigroups (of size $n^n$), and the largest number of atoms ($2^n$), all of which have the maximal possible state complexities. 
In this sense these are \emph{most complex} regular languages.

In this paper we show that there also exist most complex non-returning languages.
For each $n \ge 4$, we define a language of state complexity $n$.
We prove that the syntactic semigroup of this language has $(n-1)^n$ elements,  that it is generated by $\binom{n}{2}$ elements, and that the number of generators cannot be reduced. 
We also show that this language has $2^n$ atoms, all of which have maximal state complexity.
We demonstrate that the  upper bound on the state complexity of  reversal is met by a single ternary language, and  that no binary language meets this bound. Moreover,  restrictions of this language to binary alphabets meet the bounds for star, product and boolean operations.
This is in contrast to~\cite{EHJ16} where several types of witnesses are used to meet the various bounds.
Thus we correct an error in~\cite[Table 1]{EHJ16}, where it is stated that the upper bound on the complexity of product cannot be reached with binary witnesses.
Finally, we also  derive upper bounds on \emph{unrestricted} binary operations~\cite{Brz16}, that is, operations on languages over different alphabets, and show that our non-returning language meets these bounds. 

\section{Preliminaries}
\label{sec:preliminaries}

A \emph{deterministic finite automaton (DFA)} is a quintuple
$\mathcal{D}=(Q, \Sigma, \delta, q_0,F)$, where
$Q$ is a finite non-empty set of \emph{states},
$\Sigma$ is a finite non-empty \emph{alphabet},
$\delta\colon Q\times \Sigma\to Q$ is the \emph{transition function},
$q_0\in Q$ is the \emph{initial} state, and
$F\subseteq Q$ is the set of \emph{final} states.
We extend $\delta$ to a function $\delta\colon Q\times \Sigma^*\to Q$ as usual.
A~DFA $\mathcal{D}$ \emph{accepts} a word $w \in \Sigma^*$ if ${\delta}(q_0,w)\in F$. The language accepted by $\mathcal{D}$ is denoted by $L(\mathcal{D})$. If $q$ is a state of $\mathcal{D}$, then the language $L^q$ of $q$ is the language accepted by the DFA $(Q,\Sigma,\delta,q,F)$. 
A state is \emph{empty} if its language is empty. Two states $p$ and $q$ of $\mathcal{D}$ are \emph{equivalent} if $L^p = L^q$. 
A state $q$ is \emph{reachable} if there exists $w\in\Sigma^*$ such that $\delta(q_0,w)=q$.
A DFA is \emph{minimal} if all of its states are reachable and no two states are equivalent.

A \emph{nondeterministic finite automaton (NFA)} is a quintuple
$\mathcal{D}=(Q, \Sigma, \delta, I,F)$ where
$Q$,
$\Sigma$ and $F$ are defined as in a DFA, 
$\delta\colon Q\times \Sigma\to 2^Q$ is the \emph{transition function}, and
$I\subseteq Q$ is the \emph{set of initial states}. 

We use $Q_n=\{0,\dots,n-1\}$ as our basic set with $n$ elements.
A \emph{transformation} of $Q_n$ is a mapping $t\colon Q_n\to Q_n$.
The \emph{image} of $q\in Q_n$ under $t$ is denoted by $qt$, and this notation is extended to subsets of $Q_n$.
The \emph{rank} of a transformation $t$ is the cardinality of $Q_nt$.
If $s$ and $t$ are transformations of $Q_n$, their composition is  denoted $(qs)t$ when applied to $q \in Q_n$.
Let $\mathcal{T}_{Q_n}$ be the set of all $n^n$ transformations of $Q_n$; then $\mathcal{T}_{Q_n}$ is a monoid under composition. 

For $k\ge 2$, a transformation 
$t$ of a set $P=\{q_0,q_1,\ldots,q_{k-1}\} \subseteq Q_n$ is a \emph{$k$-cycle}
if $q_0t=q_1, q_1t=q_2,\ldots,q_{k-2}t=q_{k-1},q_{k-1}t=q_0$.
This $k$-cycle is denoted by $(q_0,q_1,\ldots,q_{k-1})$, and leaves the states in $Q_n\setminus P$ unchanged.
A~2-cycle $(q_0,q_1)$ is called a \emph{transposition}.
A transformation  that sends state $p$ to $q$ and acts as the identity on the remaining states is denoted by $(p \to q)$.
If a transformation of $Q_n$ has rank $n-1$, then there is exactly one pair of distinct elements $i,j \in Q_n$ such that $it = jt$.
We say a transformation $t$ of $Q_n$ is of \e{type $\{ i, j \}$} if $t$ has rank $n-1$ and $it = jt$ for $i < j$.

The 
\emph{syntactic congruence} of a language $L\subseteq \Sigma^*$ is defined on $\Sigma^+$ as follows:
For $x, y \in \Sigma^+,  x \,{\mathbin{\approx_L}}\, y $  if and only if  $wxz\in L  \Leftrightarrow wyz\in L$ for all  $w,z \in\Sigma^*.
$
The quotient set $\Sigma^+/ {\mathbin{\approx_L}}$ of equivalence classes of  ${\mathbin{\approx_L}}$ is a semigroup, the \emph{syntactic semigroup} $T_L$ of $L$.

Let $\mathcal{D} = (Q_{n}, \Sigma, \delta,  0, F)$ be a DFA. For each word $w \in \Sigma^*$, the transition function induces a transformation $\delta_w$ of $Q_n$ by  $w$: for all $q \in Q_n$, 
$q\delta_w = \delta(q, w).$ 
The set $T_{\mathcal{D}}$ of all such transformations by non-empty words is the \emph{transition semigroup} of $\mathcal{D}$ under composition~\cite{Pin97}. 
Sometimes we use the word $w$ to denote the transformation it induces; thus we write $qw$ instead of $q\delta_w$.
We extend the notation to sets: if $P\subseteq Q_n$, then $Pw=\{pw\mid p\in P\}$.
We also find  it convenient to write $P\stackrel{w}{\longrightarrow} Pw$ to indicate that the image of $P$ under $w$ is $Pw$.

If  $\mathcal{D}$ is a minimal DFA of $L$, then $T_{\mathcal{D}}$ is isomorphic to the syntactic semigroup $T_L$ of $L$~\cite{Pin97}, and we represent elements of $T_L$ by transformations in~$T_{\mathcal{D}}$. 
The size of this semigroup has been used as a measure of complexity~\cite{Brz13,BrYe11,HoKo04,KLS05}.

Atoms are defined by a left congruence, where two words $x$ and $y$ are equivalent whenever 
 $ux\in L$ if and only if  $uy\in L$ for all $u\in \Sigma^*$. 
 Thus $x$ and $y$ are equivalent whenever  $x\in u^{-1}L$ if and only if $y\in u^{-1}L$
 for all $u\in \Sigma^*$.
 An equivalence class of this relation is an \emph{atom} of $L$~\cite{BrTa14}. 
Thus an atom is a non-empty intersection of complemented and uncomplemented quotients of $L$. 
The number of atoms and their  state complexities were suggested as possible measures of complexity of regular languages~\cite{Brz13},
because all the quotients of a language, and also the quotients of atoms, are always unions of atoms
~\cite{BrTa13,BrTa14,Iva16}.

Suppose $\circ$ is a unary operation on languages, and $f(n)$ is an upper bound on the state complexity of this operation. 
If the state complexity of $(L_n)^\circ$ is $f(n)$, then  $L_n$  is called a \emph{witness} to the state complexity of $\circ$ for that $n$.
In general, we need a sequence $(L_k, L_{k+1},\dots, )$ of such languages; this sequence is called a \emph{stream}. 
 Often a stream does not start at 1 because the bound may not hold for small values of $n$.
For a binary operation we need two streams. The languages in a stream usually have the same form and differ only in the parameter $n$.

Sometimes the same stream can be used for both operands of a binary operation, but this is not always possible. For example,  for boolean operations when $m=n$, the state complexity of 
$L_n\cup L_n=L_n$ is $n$, whereas it should be $mn=n^2$.
However, in many cases the second language  is a "dialect" of the first, that is, it ``differs only slightly'' from the first. 
A \emph{dialect}  of $L_n(\Sig)$ is a language obtained from $L_n(\Sig)$  by  deleting some letters of $\Sigma$ in the words of $L_n(\Sig)$ -- by this we mean that words containing these letters are deleted -- or replacing them by letters of another alphabet $\Sigma'$.
In this paper we will encounter only two types of dialects:
\be
\item
A dialect in which some letters were deleted; for example, $L_n(a,b)$ is a dialect of $L_n(a,b,c)$ with $c$ deleted, and $L_n(a,-,c)$ is a dialect with $b$ deleted.
\item
 A dialect in which the roles of two letters are exchanged; for example, $L_n(b,a)$ is such a dialect of $L_n(a,b)$.
\ee
These two types of dialects can be combined, for example, in $L_n(a,-,b)$ the letter $c$ is deleted, and $b$ plays the role that $c$ played originally.
The notion of dialects also extends to DFAs; for example, if $\cD_n(a,b,c)$ recognizes $L_n(a,b,c)$ then $\cD_n(a,-,b)$ recognizes the dialect $L_n(a,-,b)$.

\section{Main Results}

From now on by \emph{complexity} we mean \emph{state complexity}.
If a word $w$ induces a transformation $t$ in a DFA, we denote this by $w\colon t$.

Let $\Gamma =\{a_{i,j} \mid 0\le i < j \le n-1\}$, where $a_{i,j}$ is a letter that induces any transformation of type $\{i,j\}$ and does not map any state to $0$. Let 
$\Gamma' = \Gamma \setminus \{ a_{0,n-1}, a_{0,1}, a_{1,n-1}, a_{1,2}\}$.
Let $\Sig = \{a,b,c,d\} \cup \Gamma'  $, where 
$a: (1,\dots,n-1) (0 \to 1)$,
$b:  (1,2) (0 \to 2)$, 
$c:  (2,\dotsc,n-1)(1 \to 2)(0 \to 1)$, and 
$d: (0 \to 2)$.
That is:
\bi
\item
$qa = q+1$ for $0 \le q \le n-2$, and $(n-1)a = 1$.
\item
$0b = 2$, $1b = 2$, $2b = 1$ and $qb = q$ for $q \not\in \{0,1,2\}$.
\item
$qc = q+1$ for $0 \le q \le n-2$, and $(n-1)c = 2$.
\item
$0d = 2$ and $qd = q$ for $q \ne 0$.
\ei
Note that $a$, $b$, $c$ and $d$ are transformations of types $\{0, n-1\}$, $\{0,1\}$, $\{1, n-1\}$ and  $\{0,2\}$, respectively.
Note also that $a$, $b$ and $c$ restricted to $Q_n\setminus \{0\}$ generate all the transformations of $\{1,\dots,n-1\}$.
This follows from the well-known fact that the full transformation semigroup on a set $X$ can be generated by the symmetric group on $X$ together with a transformation of $X$ with rank $|X|-1$. For $X = \{1,\dotsc,n-1\}$, we see that $\{(1,\dotsc,n-1),(1,2)\}$ (the restrictions of $a$ and $b$) generate the symmetric group, and $(2,\dotsc,n-1)(1 \ra 2)$ (the restriction of $c$) is a transformation of rank $|X|-1 = n-2$.

We are now ready to define a most complex non-returning DFA and language.

\begin{definition}
\label{def:most_complex}
For $n\ge 4$, let $\mathcal{D}_n=\mathcal{D}_n(\Sig)=(Q_n,\Sigma,\delta_n, 0, \{n-1\})$, where 
$\Sigma=\{a,b,c,d\} \cup \Gamma'$,
and $\delta_n$ is defined in accordance with the transformations described above.
 See Figure~\ref{fig:non_returning}.
Let $L_n=L_n(\Sig)$ be the language accepted by~$\mathcal{D}_n(\Sig)$.
\end{definition}

\begin{figure}[ht]
\unitlength 7.5pt
\begin{center}\begin{picture}(39,7)(0,4)
\gasset{Nh=1.8,Nw=3.5,Nmr=1.25,ELdist=0.4,loopdiam=1.5}
	{\scriptsize
\node(0)(1,7){0}\imark(0)
\node(1)(7,7){1}
\node(2)(15,7){2}
\node(3)(21,7){3}
}
\node[Nframe=n](3dots)(27,7){$\dots$}
	{\scriptsize
\node(n-2)(33,7){$n-2$}
	}
	{\scriptsize
\node(n-1)(39,7){$n-1$}\rmark(n-1)
	}
\drawedge(0,1){$a,c$}
\drawloop(1){$d$}
\drawedge[curvedepth=5,ELdist=.1](0,2){$b,d$}
\drawedge[curvedepth= .8,ELdist=.1](1,2){$a,b,c$}
\drawedge[curvedepth= .8,ELdist=-1.2](2,1){$b$}
\drawloop(2){$d$}
\drawedge(2,3){$a,c$}
\drawloop(3){$b,d$}
\drawedge(3,3dots){$a,c$}
\drawedge(3dots,n-2){$a,c$}
\drawloop(n-2){$b,d$}
\drawedge(n-2,n-1){$a,c$}
\drawedge[curvedepth= 4.0,ELdist=-1.0](n-1,1){$a$}
\drawedge[curvedepth= 2.0,ELdist=-1.0](n-1,2){$c$}
\drawloop(n-1){$b,d$}
\end{picture}\end{center}
\caption{Most complex non-returning language $L_n(\Sig)$ of Definition \ref{def:most_complex}. The letters in $\Gamma' = \Sig \setminus \{a,b,c,d\}$ are omitted.}
\label{fig:non_returning}
\end{figure}
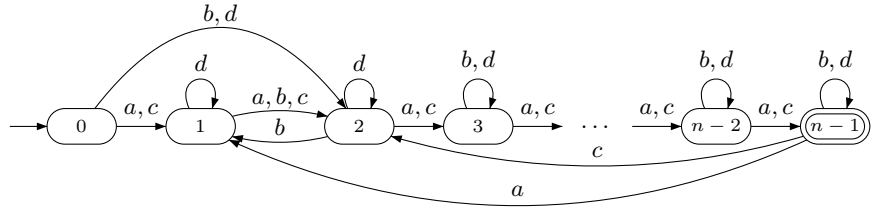

\begin{theorem}[Most Complex Non-Returning Languages]
\label{thm:leftideals}
For each $n\ge 4$, the DFA of Definition~\ref{def:most_complex} is minimal and non-returning. 
The stream $(L_n(\Sig) \mid n \ge 4)$  with some dialect streams
is most complex in the class of regular non-returning languages in the following sense:
\begin{enumerate}
\item
The syntactic semigroup of $L_n(\Sig)$ has cardinality $(n-1)^n$, and at least 
$\binom{n}{2}$ letters are required to reach this bound.
\item
Each quotient of $L_n(a)$ has complexity $n-1$, except $L$ itself, which has  complexity $n$.
\item
The reverse of $L_n(a,b,c)$ has complexity $2^n$,  and at least three letters are needed to meet this bound.
Moreover, $L_n(a,b,c)$ has $2^n$ atoms. 
\item
For each atom $A_S$ of   $L_n(\Sig)$, the complexity $\kappa(A_S)$ satisfies:\\
\begin{equation*}
	\kappa(A_S) \le
	\begin{cases}
		2^{n-1}, 			& \text{if $S\in \{\emp,Q_n\}$;}\\
		2+ \sum_{x=1}^{|S|} \sum_{y=1}^{|S|} \binom{n-1}{x}\binom{n-1-x}{y},
		 			& \text{if $\emp \subsetneq S \subsetneq Q_n$.}
		\end{cases}
\end{equation*}
 Moreover, at least $\binom{n}{2}$ letters are required to meet these bounds.

\item
The star of $L_n(a,b)$ has complexity $2^{n-1}$.
\item
	\begin{enumerate}
	\item
	Restricted  product:
	$\kappa(L_m(a,b) L_n(a,-,b)) = (m-1)2^{n-1}+1$.
	\item
	Unrestricted  product:
	For $m,n \ge 4$, let $L_m$ (respectively, $L_n$) be a non-returning language of complexity $m$ (respectively, $n$)  over an alphabet $\Sig'$, (respectively, $\Sig$). 
Then the complexity of product is at most $m2^{n-1}+1$, and this bound is met by
 $L_m(a,b)$ and $L_n(a,-,b,d)$.

	\end{enumerate}
\item
	\begin{enumerate}
	\item
	Restricted boolean operations:
	$\kappa(L_m(a,b)\circ L_n(b,a)) = mn-(m+n-2)$.
	\item
	In case $m\neq n$, $\kappa(L_m(a,b)\circ L_n(a,b)) = mn-(m+n-2)$.
	\item
	Unrestricted boolean operations:
	The complexity of $L_m(a,b,c) \circ L_n(b,a,d)$ 
	is $mn+1$ if $\circ\in \{\cup,\oplus\}$,
	that of  $L_m(a,b,c) \setminus L_n(b,a)$ is 
 	$mn-n+2$, and 
 	 of $L_m(a,b) \cap L_n(b,a)$ is $mn-(m+n-2)$. 
	\end{enumerate}
\end{enumerate}
All of these bounds are maximal for non-returning languages.
\end{theorem}
\begin{proof}
From the definition of the letters of $\Sig$ it is obvious that the DFA $\cD_n$ is non-returning, and that any pair $(p,q)$ of states can be distinguished by the shortest word in $a^*$ accepted by $p$ but not by $q$.
\be
\item
This follows from Propositions~\ref{prop:necessary} and~\ref{prop:sufficient}.
\item
Obvious.
\item
By Proposition~\ref{prop:atomnum} the number of atoms of $L_n(a,b,c)$ is $2^n$. By~\cite{BrTa14} the complexity of the reverse is the same as the number of atoms.
 By Proposition~\ref{prop:3letters} at least three letters are required to meet this bound on the number of atoms and the complexity of reverse.
\item
See Propositions~\ref{prop:atombounds}, \ref{prop:atomboundsmet}, and ~\ref{prop:alphabet_atom_bounds}.
\item
See Proposition~\ref{prop:star}.
\item
See Propositions~\ref{prop:restricted_product} and~\ref{prop:unrestricted_product}.
\item
See Propositions~\ref{prop:restricted_boolean} 
and~\ref{prop:unrestricted_boolean}.
\ee
The rest of this paper is devoted to the proofs of the propositions  above.
\qed
\end{proof}

\section{Syntactic Semigroup}
For all basic operations on non-returning languages, the complexity bounds can be met with either binary or ternary witnesses~\cite{EHJ16}. However, to meet the bound for  the size of the syntactic semigroup, our most complex stream is forced to use an alphabet that grows quadratically in size.

Let $Q_n = \{0,1,\dotsc,n-1\}$.
For $n \ge 2$, let $N_n$ denote the semigroup of transformations of $Q_n$ such that $it \ne 0$ for all $i \in Q_n$.
We call $N_n$ the \e{full non-returning semigroup} of degree $n$.
We will show that a minimal generating set for $N_n$ must have size $\binom{n}{2}$.

\begin{proposition}
\label{prop:necessary}
If $G$ is a generating set for $N_n$, then $G$ contains a transformation of 
type $\{ i,j \}$ for each set $\{i,j\} \subseteq Q_n$.
In particular, a minimal generating set has exactly one element of type $\{i,j\}$ for each $\{i,j\} \subseteq Q_n$, and there are $\binom{n}{2}$ such elements.
\end{proposition}

\begin{proof}
Suppose $t$ is a transformation of type $\{ i,j \}$, and let $t'$ be an arbitrary transformation. If $tt'$ has rank $n-1$, then $tt'$ has type $\{ i,j \}$. Indeed, since $it = jt$, it follows that $itt' = jtt'$.
Thus composing a transformation of type $\{ i,j \}$ with an arbitrary transformation either preserves the type, or lowers the rank.

Now, suppose $G$ generates $N_n$. 
Observe that $N_n$ does not contain transformations of rank $n$, since such transformations necessarily map some element to $0$.
Since composition with a transformation of type $\{ i,j \}$ either preserves type or lowers rank, the semigroup generated by $G$  contains only transformations that either have the same type as some element of $G$, or have rank less than $n-1$ and so are typeless.
But $N_n$ contains a transformation of type $\{ i,j \}$ for each $\{i,j\} \subseteq Q_n$. 
So if $G$ generates $N_n$, then $G$ must contain an element of type $\{ i,j \}$ for each $\{i,j\} \subseteq Q_n$. \qed
\end{proof}

This gives a necessary condition for $G$ to generate $N_n$. Now we give a sufficient condition.

\begin{proposition}
\label{prop:sufficient}
Let $G$ be a subset of $N_n$ that contains a transformation of type $\{ i,j \}$ for each set $\{i,j\} \subseteq Q_n$.
Let $G'$ be obtained by restricting every transformation in $G$ to $Q_n \setminus \{0\}$.
If $G'$ generates the symmetric group on $Q_n \setminus \{0\}$, then $G$ generates $N_n$.
\end{proposition}
\begin{proof}
First, we show that $G'$  in fact generates the full transformation semigroup on $Q_n \setminus \{0\}$.
Recall that the full transformation semigroup on $X$ is generated by the symmetric group on $X$ together with a transformation of $X$ of rank $|X|-1$.
By assumption, $G'$ contains generators of the symmetric group on $Q_n \setminus \{0\}$.
Transformations in $G$ of type $\{ i,j \}$ with $0 < i,j$ have rank $n-1$, and furthermore their restrictions  to $Q_n\setminus \{0\}$  are of rank $n-2$.
Thus $G'$ contains generators of the symmetric group on $Q_n \setminus \{0\}$, as well as a transformation of rank $|Q_n\setminus \{0\}|-1 = n-2$; it follows that $G'$ generates the full transformation semigroup on $Q_n\setminus\{0\}$.

Now, we prove that $G$ generates every transformation in $N_n$.
Let $t$ be an element of $N_n$; we want to show that $t$ is in the semigroup generated by $G$.
Since $N_n$ does not contain any transformations of rank $n$, the transformation $t$ has rank less than $n$, and thus there exist distinct $i,j \in Q_n$ such that $it = jt$.
Select a transformation $s$ of type $\{ i,j \}$ in $G$.
Then for distinct $q,q' \in Q_n$, we have $qs = q's$ if and only if $\{q,q'\} = \{i,j\}$.

Hence there is a well-defined transformation $r'$ of $Q_n \setminus \{0\}$ given by $(qs)r' = qt$ for all $q \in Q$; it is well-defined since if we have $qs = q's$, then $\{q,q'\} = \{i,j\}$ and $is$ and $js$ get mapped to a common element $it = jt$.
The transformation $r'$ lies in the full transformation semigroup on $Q_n \setminus \{0\}$, and so it is in the semigroup generated by $G'$.
Hence there is some transformation $r$ of $Q_n$ in the semigroup generated by $G$ such that $r$ is equal to $r'$ when restricted to $Q_n \setminus \{0\}$.

Since $qs \in Q_n \setminus \{0\}$ for all $q \in Q_n$, it follows that $(qs)r = (qs)r' = qt$ for all $q \in Q_n$, and thus $sr$ and $t$ are equal as transformations.
Since $s$ is in $G$ and $r$ is in the semigroup generated by $G$, it follows $sr = t$ is in the semigroup generated by $G$.
Thus the semigroup generated by $G$ contains all elements of $N_n$; but $G$ is a subset of $N_n$, so $G$ generates $N_n$. \qed
\end{proof}

\section{Number and Complexities of Atoms}

Denote the complement of a language $L$ by $\ol{L} = \Sig^* \setminus L$.
Let $Q_n=\{0,\dots,n-1\}$ and let $L_n$ be a non-empty regular language with quotients $K = \{K_0,\dotsc,K_{n-1}\}$. Each subset $S$ of $Q_n$ defines an \emph{atomic intersection} $A_S = \bigcap_{i \in S} K_i \cap \bigcap_{i \in \ol{S}} \ol{K_i}$, where $\ol{S} = Q_n \setminus S$.
An \emph{atom} of $L$ is a non-empty atomic intersection; this definition is equivalent to that given in Section~\ref{sec:preliminaries} in terms of a left congruence. 
Note that if $S \ne T$, then $A_S \cap A_T = \emp$; hence atoms corresponding to distinct subsets of $Q_n$ are disjoint.
A language of complexity $n$ can have at most $2^n$ atoms, since there are $2^n$ subsets of $Q$. We show that this bound can be met by non-returning languages.
Additionally, we derive upper bounds on the complexities of atoms of non-returning languages, and show that our most complex stream meets these bounds.

We now describe a  construction due to Iv\'an~\cite{Iva16}.
Let $L$ be a regular language with DFA $\cD = (Q,\Sig,\delta,q_0,F)$. For each $S \subseteq Q$, we define a DFA $\cD_S = (Q_S,\Sig,\Delta,(S,\ol{S}),F_S)$ as follows. 
\bi
\item
$Q_S = \{(X,Y) : X,Y \subseteq Q, X \cap Y = \emp\} \cup \{\bot\}$. State $\bot$ is the \emph{sink state}.
\item
$\Delta((X,Y),a) = (Xa,Ya)$ if $Xa \cap Ya = \emp$, and otherwise $\Delta((X,Y),a) = \bot$; also $\Delta(\bot,a) = \bot$.
\item
$F_S = \{(X,Y): X \subseteq F, Y \subseteq \ol{F}\}$.
\ei
The DFA $\cD_S$ recognizes the atomic intersection $A_S$ of $L$; if it recognizes a non-empty language, then $A_S$ is an atom. We can determine the complexity of $A_S$ by counting reachable and distinguishable states in $\cD_S$.

\subsection{Number of Atoms}

\begin{proposition}
\label{prop:atomnum}
The language $L_n = L_n(a,b,c)$ has $2^n$ atoms.
\end{proposition}
\begin{proof}
We want to show that $A_S$ is an atom of $L_n$ for all $S \subseteq Q_n$.
It suffices to show for each $S$ that the DFA $\cD_S$ recognizes at least one word.
Then since atoms corresponding to different subsets of $Q_n$ are disjoint, this proves there are $2^n$ distinct atoms.

First, we show that from the initial state $(S,\ol{S})$, we can reach some state of the form $(X,Y)$ where $0 \not\in X$ and $0 \not\in Y$.
Consider the set $\{0,1,n-1\}$. Notice that for each subset $\{i,j\}$ of $\{0,1,n-1\}$, we have a transformation of type $\{i,j\}$: $a$ has type $\{0,n-1\}$, $b$ has type $\{0,1\}$, and $c$ has type $\{1,n-1\}$.
Additionally, by the pigeonhole principle, either $S$ contains two distinct elements from $\{0,1,n-1\}$, or $\ol{S}$ contains two distinct elements from $\{0,1,n-1\}$.

Suppose without loss of generality it is $S$ which contains two distinct elements from $\{0,1,n-1\}$.
Let $\{i,j\} \subseteq S$ for some $\{i,j\} \subseteq \{0,1,n-1\}$ with $i \ne j$.
Let $\sig \in \Sig$ be the letter inducing the transformation of type $\{i,j\}$.
Then we claim $(S,\ol{S})\sig \ne \bot$.
Indeed, suppose that $q \in S\sig \cap \ol{S}\sig$.
Then since $\sig$ is a transformation of type $\{i,j\}$, we must have $i\sig = j\sig = q$, and no other element is mapped to $q$.
But $\{i,j\} \subseteq S$, so we cannot have $q \in \ol{S}\sig$.

Hence $S\sig \cap \ol{S}\sig = \emp$.
Furthermore, since $\sig$ is a non-returning transformation, we have $0 \not\in S\sig$ and $0 \not\in \ol{S}\sig$.
Thus starting from the initial state $(S,\ol{S})$, we can apply $\sig$ to reach a state of the form $(X,Y)$ with $0 \not\in X$ and $0 \not\in Y$.

Now, recall that the three transformations $\{a,b,c\}$, when restricted to $Q_n \setminus \{0\}$, generate all transformations of $Q_n \setminus \{0\}$. Since $X \subseteq Q_n \setminus \{0\}$, there exists a transformation of $Q_n \setminus \{0\}$ that maps every element of $X$ to $n-1$ and every element of $(Q_n \setminus \{0\}) \setminus X$ to $1$. Let $w \in \{a,b,c\}^*$ be a word that induces this transformation when restricted to $Q_n \setminus \{0\}$. Since $Y \subseteq Q_n \setminus \{0\}$ and $Y$ is disjoint from $X$, it follows that $w$ maps every element of $Y$ to $1$.
Since $F_n = \{n-1\}$ is the final state set of $\cD_n$, we see that $Xw \subseteq F_n$ and $Yw \subseteq \ol{F_n}$. Thus $(Xw,Yw) = (\{n-1\},\{1\})$ is a final state of $\cD_S$.

This shows that there exists a word $\sig w \in \{a,b,c\}^*$ that maps the initial state $(S,\ol{S})$ of $\cD_S$ to a final state. Thus  $A_S$ is an atom. \qed
\end{proof}

Next, we prove that the bound on number of atoms cannot be met by a binary witness. 
From~\cite{BrTa14} we know that the number of atoms of a regular language is equal to the state complexity of the reversal of the language. 
Hence this also proves a conjecture from~\cite{EHJ16}, that a ternary witness is necessary to meet the bound for reversal of non-returning languages. 

\begin{proposition}
\label{prop:3letters}
Let $L$ be a non-returning language of complexity $n$ over $\Sig = \{a,b\}$. Then the number of atoms of $L$ is strictly less than $2^n$.
\end{proposition}
\begin{proof}
Let $\cD$ be the minimal DFA of $L$, with state set $Q_n$.
We introduce some special terminology for this proof, which generalizes the notion of transformations of type $\{i,j\}$.
We say that a transformation $t$ \e{unifies} $i$ and $j$, or unifies the set $\{i,j\}$, if $it = jt$.
For example, transformations of type $\{i,j\}$ unify $\{i,j\}$.
But furthermore, every transformation of $Q_n$ of rank $n-1$ or less unifies at least one pair of elements of $Q_n$.
The transition semigroup of $\cD$ cannot have transformations of rank $n$, since $L$ is non-returning; thus all the transformations in the transition semigroup must unify some pair of states.

Suppose that in $\cD$, the letter $a$ induces a transformation that unifies $\{i,j\}$, and $b$ induces a transformation that unifies $\{k,\ell\}$.
Assume also that $i \ne j$ and $k \ne \ell$.
We will show that at least one atomic intersection $A_S$ of $L$ is empty, and thus is not an atom.

Suppose $\{i,j\} = \{k,\ell\}$.
Let $S = \{i\}$ and consider the atomic intersection $A_S$.
The initial state of the DFA for $A_S$ is $(\{i\},\ol{S})$.
Note that $j \in \ol{S}$, so $ja \in \ol{S}a$.
But $a$ unifies $i$ and $j$, so $ja = ia \in \{i\}a$.
Thus since $\{i\}a \cap \ol{S}a \ne \emp$, the letter $a$ sends the initial state $(\{i\},\ol{S})$ to the sink state.
Since $b$ also unifies $i$ and $j$, the letter $b$ also sends $(\{i\},\ol{S})$ to the sink state. 
Thus $A_S$ is non-empty if and only if $(\{i\},\ol{S})$ is a final state.
In fact, either $A_S$ is non-empty or $A_S = \{\eps\}$, since every non-empty word sends the initial state $(\{i\},\ol{S})$ to the sink state.
If we let $T = \{j\}$, the same argument shows that $A_T$ is either empty or $A_T = \{\eps\}$.
But $A_S \cap A_T = \emp$, so one of $A_S$ or $A_T$ must be empty.

Now, suppose $\{i,j\} \cap \{k,\ell\} = \emp$.
Let $S = \{i,k\}$ and consider the atomic intersection $A_S$.
The initial state of the DFA for $A_S$ is $(\{i,k\},\ol{S})$ with $j,\ell \in \ol{S}$.
Thus similarly to before, the transformation $a$ which unifies $\{i,j\}$ and the transformation $b$ which unifies $\{k,\ell\}$ both send $A_S$ to a sink state.
So either $A_S$ is empty or $A_S = \{\eps\}$.
For $T = \{j,\ell\}$, the same argument shows that either $A_T$ is empty or $A_T = \{\eps\}$.
Hence as before, one of $A_S$ or $A_T$ is empty.

Finally, suppose $\{i,j\} \cap \{k,\ell\}$ has exactly one element.
Then either $k \in \{i,j\}$ or $\ell \in \{i,j\}$.
Assume without loss of generality that $\ell \in \{i,j\}$ and $\ell = i$; otherwise rename the elements so this is the case.
Then $a$ unifies $\{i,j\}$, and $b$ unifies $\{i,k\}$.
Let $S = \{i\}$ and consider $A_S$.
As before, the initial state of the DFA for $A_S$ is sent to a sink state by both $a$ and $b$.
Thus either $A_S$ is empty or $A_S = \{\eps\}$.
For $T = \{j,k\}$, the same argument shows that either $A_T$ is empty or $A_T = \{\eps\}$.
Hence one of $A_S$ or $A_T$ is empty.
\qed
\end{proof}

\subsection{Complexity of Atoms}
First, we derive upper bounds for the complexity of atoms of non-returning languages.
\begin{proposition}
\label{prop:atombounds}
Let $L$ be a non-returning language of complexity $n$, and let $Q_n$ be the state set of its minimal DFA. Let $S \subseteq Q_n$; then we have
\begin{equation*}
	\kappa(A_S) \le
	\begin{cases}
		2^{n-1}, 			& \text{if $S\in \{\emp,Q_n\}$;}\\
		2+ \sum_{x=1}^{|S|} \sum_{y=1}^{|S|} \binom{n-1}{x}\binom{n-1-x}{y},
		 			& \text{if $\emp \subsetneq S \subsetneq Q_n$.}
		\end{cases}
\end{equation*}
\end{proposition}
\begin{proof}
To obtain an upper bound on the complexity of the atomic intersection $A_S$, we find an upper bound on the number of reachable states in the DFA $\cD_S$.

Let $S = \emp$; then the initial state of $\cD_S$ is $(\emp,Q_n)$. From this state, we can only reach states of the form $(\emp,X)$ with $X \subseteq Q_n \setminus \{0\}$ and $X$ non-empty. The fact that $0 \not\in X$ follows from the non-returning property. There are $2^{n-1}-1$ non-empty subsets of $Q_n \setminus \{0\}$; adding 1 for the initial state gives the upper bound of $2^{n-1}$. For $S = Q_n$ we use a symmetric argument. Note that if $S = \emp$ or $S = Q_n$, the sink state $\bot$ is not reachable in $\cD_S$.

Now suppose $S \ne \emp$ and $S \ne Q_n$. The initial state of $\cD_S$ is $(S,\ol{S})$. By the non-returning property, this is the only state $(X,Y)$ of $\cD_S$ that has $0 \in X$ or $0 \in Y$. Hence all non-initial, non-sink reachable states have the form $(X,Y)$ where $1 \le |X| \le |S|$, $1 \le |Y| \le |\ol{S}| = n-|S|$, $X \cap Y = \emp$, $0 \not\in X$, and $0 \not\in Y$.

Suppose $|X| = x$. There are $\binom{n-1}{x}$ subsets of $Q_n \setminus \{0\}$ of size $x$.
Given a subset $X$, if we want to choose a set $Y \subseteq Q_n \setminus \{0\}$ which has size $y$ and is disjoint from $X$, there are $\binom{n-1-x}{y}$ choices: we must choose $y$ elements from $Q_n \setminus \{0\}$, but we cannot choose any of the $x$ elements from $X$, so there are really $n-1-x$ options to choose from. Hence the number of non-initial, non-sink reachable states $(X,Y)$ with $|X| = x$ and $|Y| = y$ is bounded by
\[ \binom{n-1}{x}\binom{n-1-x}{y}. \]
Since $|X|$ can range from $1$ to $|S|$ and $|Y|$ can range from $1$ to $n-|S|$, we take the sum over thse ranges to get a bound on the total number of non-initial, non-sink reachable states:
\[ \sum_{x=1}^{|S|}\sum_{y=1}^{n-|S|} \binom{n-1}{x}\binom{n-1-x}{y}. \]
Finally, add $2$ for the initial and sink states. This gives the stated bound. \qed
\end{proof}

Next, we prove our witness meets these bounds. First, we recall a lemma from~\cite{BrDa15}.
\begin{lemma}
\label{lem:dist}
Let $L$ be a regular language, and let $\cD_S$ be the DFA of the atomic intersection $A_S$ of $L$. Then:
\be
\item
States $(X,Y)$ and $(X',Y')$ of $\cD_S$ are distinguishable 
if $X \ne X'$ and $A_{X},A_{X'}$ are both atoms of $L$, or if $Y \ne Y'$ and $A_{\ol{Y}},A_{\ol{Y'}}$ are both atoms of $L$.
\item
If one of $A_X$ or $A_{\ol{Y}}$ is an atom of $L$, then $(X,Y)$ is distinguishable from $\bot$.
\ee
\end{lemma}

\begin{proposition}
\label{prop:atomboundsmet}
The atoms of the language $L_n = L_n(\Sig)$ meet the complexity bounds of Proposition \ref{prop:atombounds}.
\end{proposition}
\begin{proof}
By Proposition \ref{prop:atomnum}, $A_S$ is an atom of $L_n$ for each $S \subseteq Q_n$.
Thus by Lemma \ref{lem:dist}, for each $S \subseteq Q_n$, all distinct states of $\cD_S$ are distinguishable.
So it suffices to just show that the number of reachable states of each atom $A_S$ meets the bound.
Recall that the transformations in $\Sig$ generate the full non-returning semigroup $N_n$, which contains all transformations of $Q_n$ that do not map anything to $0$.

Suppose $S = \emp$. The initial state of $\cD_S$ is $(\emp,Q_n)$. We claim that $(\emp,X)$ is reachable for all non-empty $X \subseteq Q_n \setminus \{0\}$. Indeed, just let $w \in \Sig^*$ induce a transformation that maps $Q_n$ onto $X$; then $(\emp,Q_n)w = (\emp,X)$. Hence there are $2^{n-1}$ reachable states. A symmetric argument works for $S = Q_n$.

Suppose $S \ne \emp$ and $S \ne Q_n$. We see from the proof of Proposition \ref{prop:atombounds} that to show that the number of reachable states of $\cD_S$ meets the bound, it suffices to show that every state $(X,Y)$ with $1 \le |X| \le |S|$, $1 \le |Y| \le |\ol{S}|$, $X \cap Y$, $0 \not\in X$, and $0 \not\in Y$ is reachable, in addition to the sink state $\bot$. The sink state can be reached from the initial state $(S,\ol{S})$ using a constant transformation. To reach $(X,Y)$, simply use a transformation that maps $S$ onto $X$ and $\ol{S}$ onto $Y$; this is possible since $X$ and $Y$ are disjoint, and also since $X,Y \subseteq Q_n \setminus \{0\}$ and we have all transformations of $Q_n$ that do not map anything to $0$. Hence the number of reachable states matches the bound. \qed
\end{proof}

We used the full alphabet of our witness in the previous proof. The following proposition shows that this was necessary.

\begin{proposition}
\label{prop:alphabet_atom_bounds}
Let $L$ be a non-returning language over $\Sig$ of complexity $n$. If the atoms of $L$ meet the bounds of Proposition \ref{prop:atombounds}, then $\Sig$ has size at least $\binom{n}{2}$.
\end{proposition}
\begin{proof}
We claim that to meet the bounds, $\Sig$ must contain a letter inducing a transformation of type $\{i,j\}$ for each $\{i,j\} \subseteq Q_n$. There are $\binom{n}{2}$ of these subsets.

To see this, fix $\{i,j\}$ and suppose no letter induces a transformation of type $\{i,j\}$. Then the transition semigroup of the minimal DFA of $L$ does not contain a transformation of type $\{i,j\}$, since there is no transformation of this type in its generating set.

Now, let $S = \{i,j\}$ and consider the atomic intersection $A_S$. The initial state of the DFA for this atomic intersection is $(\{i,j\},\ol{\{i,j\}})$. From this state, there is no way to reach a state of the form $(X,Y)$ with $|X| = 1$ and $|Y| = n-2$, since this would require a transformation of type $\{i,j\}$. Hence $A_S$ cannot have maximal complexity. \qed
\end{proof}

\section{Star}

\begin{proposition}[Star]
\label{prop:star}
Let $\cD_n(a,b)$ be the DFA of Definition~\ref{def:most_complex} and let $L_n(a,b)$ be its language. Then the complexity of $(L_n(a,b))^*$ is $2^{n-1}$.
\end{proposition}
\begin{proof}
The upper bound of $2^{n-1}$ on the complexity of star was established in~\cite{EHJ16}.
We use the construction of~\cite{EHJ16} for an NFA for $(L_n(a,b))^*$ as shown in
Figure~\ref{fig:NFA_star}: state 0 has been made final and a $b$-transition has been added from state $n-1$ to state 2.
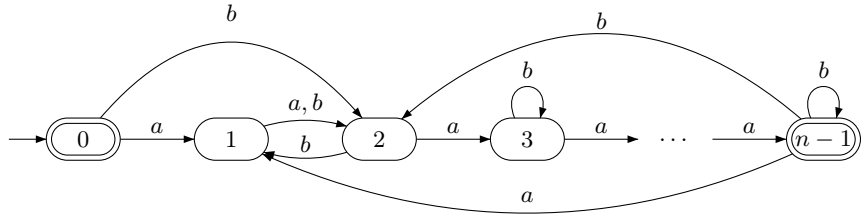
\begin{figure}[th]
\unitlength 8.0pt
\begin{center}\begin{picture}(37,10)(0,3)
\gasset{Nh=2.0,Nw=3.5,Nmr=1.25,ELdist=0.4,loopdiam=1.5}
\node(0)(1,7){0}\imark(0)\rmark(0)
\node(1)(8,7){1}\
\node(2)(15,7){2}
\node[Nframe=n](3dots)(29,7){$\dots$}
\node(3)(22,7){$3$}
	{\small
\node(n-1)(36,7){$n-1$}\rmark(n-1)
	}
\drawedge[curvedepth= 0,ELdist=.3](0,1){$a$}
\drawedge[curvedepth= 3.5,ELdist=-1](n-1,1){$a$}
\drawedge[curvedepth= -5,ELdist=-1](n-1,2){$b$}
\drawedge[curvedepth= .9,ELdist=.3](1,2){$a,b$}
\drawedge[curvedepth= .9,ELdist=-1.1](2,1){$b$}
\drawedge[curvedepth= .0,ELdist=.3](2,3){$a$}
\drawedge[curvedepth= .0,ELdist=.3](3,3dots){$a$}
\drawedge[curvedepth= 0,ELdist=.3](3dots,n-1){$a$}
\drawedge[curvedepth= 4.6,ELdist=1.0](0,2){$b$}

\drawloop(3){$b$}
\drawloop(n-1){$b$}
\end{picture}\end{center}
\caption{NFA for star of $L_n(a,b)$.}
\label{fig:NFA_star}
\end{figure}

State $\{0\}$ is initial and for $q\in\{1, \dots,n-1\}$, the set $\{q\}$ is reachable by $a^q$.

We prove by induction on the size of the set that all nonempty subsets of 
$\{1,\dots,n-1\}$ are reachable. 
We can think about the elements of the set $S=\{q_1,\dots,q_k\}$, $1\le q_1 < \cdots <q_k \le n-1$, as a pattern of states that are in the set. Applying $a$ adds  1 to each state or moves the pattern one position to the right cyclically. Applying $a^{n-2}$ subtracts 1 from each state, or moves the pattern one position to the left cyclically.
\be
\item
$q_k < n-1$.
To reach set $S=\{q_1,\dots,q_k\}$, start with set $T=\{q_2,\dots,q_k\}$, which is reachable by the induction assumption. 
Note that $q_2\ge 2$.
Apply $a^{n-1-q_k}$ so that the pattern is shifted to the right until $q_k$ reaches $n-1$.
Use $b$ to add 2. 
Now we have the required $k$ elements, but the distance between 2 and $q_2+n-1-q_k$ may be larger than that between $q_1$ and $q_2$. 
It cannot be smaller because $n-1 \ge q_k-q_1 +2$, and so $(q_2 + (n-1-q_k)) - 2 \ge q_2-q_1$.
\smallskip

If $q_2+(n-1-q_k)-2 = q_2-q_1$, then the distance between $2$ and $q_2+(n-1-q_k)$ is equal to the distance between $q_1$ and $q_2$. Thus we reach $S$ if we shift the pattern by a word in $a^*$ that sends $2$ to $q_1$.
\smallskip

Otherwise, we have $q_2+(n-1-q_k)-2 > q_2-q_1$, that is, $n-1 > q_k-q_1 +2$.
Use $a^{n-1}$ to move the pattern, including 2, to the left one position. 
Since $q_2 \ge q_1 +1$, we have $n-1 > q_k-q_2 +3$ and 
$q_2+(n-1-q_k)-1 > 2$.
Now use $b$  to move 1 to 2. This decreases the distance between 2 and 
$q_2+(n-1-q_k)$ by one. 
This can be repeated.
If we reach the pattern in which we have 2, and $q_2+(n-1-q_k)$ has been moved to 3, then we are done, because the distance between $q_2$ and $q_1$ must be at least one.
\item
$q_1 > 1$, $q_k=n-1$. A shift to the left by $a^{n-2}$ brings us to Case (1). After reaching $\{q_1-1,\dotsc,q_k-1\}$, the desired set can be reached by $a$. 
\item
$q_1=1$, $q_2=2, q_k=n-1$. 
Use $\{q_2=2, q_3, \dots,q_{k-1}, q_k =n-1\} \stackrel{b}{\longrightarrow}
\{1,2, q_3, \dots,q_{k-1}, q_k =n-1\} $.
\item
$q_1=1, q_2 > 2, q_{k-1}=n-2, q_k=n-1$. Applying $a$ we get a pattern with
$\{q_1=1, q_2=2, q_3, \dots, q_k = n-1\}$. This is reachable by Case (3). The desired pattern  is then reached by $a^{n-2}$.
\item
$q_1=1, q_2 > 2, q_{k-1} < n-2, q_k=n-1$.
We have  $\{1,q_2, \dots, n-1\} \stackrel{a}{\longrightarrow} \{1,2,q_2+1,\dots, q_{k-1}+1\}$. The latter set is reachable by Case (1), and then the desired set, by $a^{n-2}$.
\ee

For distinguishability, $\{0\}$ is the only set with $a^{n-1}$ as its shortest accepted non-empty word. If two subsets of $\{1,\dots,n-1\}$ differ by state $q$, then $a^{n-1-q}$ distinguishes them. Thus our claim holds.
\qed
\end{proof}

\section{Product}

When dealing with binary operations, to avoid confusion between the sets of states $\{0,\dots, m-1\}$ and $\{0,\dots,n-1\}$ we use 
$\mathcal{D}'_m(\Sig)=(Q'_m,\Sigma,\delta'_m, 0', \{(m-1)'\})$, 
and $\mathcal{D}_n(\Sig)=(Q_n,\Sigma,\delta_n, 0, \{n-1\})$,
where $Q'_m = \{0',\dots, (m-1)'\}$.

\begin{proposition}[Restricted Product]   
\label{prop:restricted_product}
Let $\cD_n(a,b,c)$ be the DFA of Definition~\ref{def:most_complex} and let $L_n(a,b,c)$ be its language. Then for $m,n \ge 4$  the complexity of 
$L'_m(a,b) L_n(a,-,b)$ is $(m-1)2^{n-1}+1$.
\end{proposition}
The language $L_n(a,-,b)$ is recognized by the DFA $\cD_n(a,-,b)$. This DFA has alphabet $\{a,b\}$ and transformations $a \co (1,\dotsc,n-1)(0 \to 1)$ and $b \co (2,\dotsc,n-1)(1 \to 2)(0 \to 1)$. That is, $b$ induces the same transformation in $\cD_n(a,-,b)$ that $c$ induces in $\cD_n(a,b,c)$.

We first establish a lemma.
\begin{lemma}
\label{lem:ab}
Let $a \co (1,\dotsc,n-1)(0 \to 1)$ and $b \co (2,\dotsc,n-1)(1 \to 2)(0 \to 1)$. Each element of $\{1+2k,\dotsc,n-1\}$ has exactly one inverse under the transformation $(ab)^k$, and the inverse of $q$ is $q - 2k$. In particular, if $S \subseteq \{1+2k,\dotsc,n-1\}$, then $S(ab)^{-k} = \{ q - 2k : q \in S \}$.
\end{lemma}
\begin{proof}
We proceed by induction on $k$. 
For $k = 0$, note that $(ab)^0$ is the identity transformation. Thus each $q \in \{1,\dotsc,n-1\}$ has $q-2k = q$ as its unique inverse under $(ab)^0$.

Now, assume $k \ge 1$ and the result holds for all values less than $k$.
Suppose $p(ab)^k = q$, that is, $p$ is an inverse of $q$ under $(ab)^k$.
By the induction hypothesis, $q$ has a unique inverse under $(ab)^{k-1}$, and this inverse is $q-2(k-1)$.
Thus $p(ab) = q(ab)^{-(k-1)} = q-2(k-1)$.
Now, since $1+2k\le q \le n-1$ and $k \ge 1$, we have $3 \le q-2(k-1) \le n-1$.
Since $b$ is a transformation of type $\{1,n-1\}$ and $1b = (n-1)b = 2$, the only element that does not have a unique inverse under $b$ is $2$.
Thus $q-2(k-1)$ has a unique inverse under $b$, and this inverse is $q-2(k-1)-1$.
We have $2 \le q-2(k-1)-1 \le n-2$.
Since $a$ is a permutation when restricted to $\{1,\dotsc,n-1\}$, it follows that $q-2(k-1)-1$ has a unique inverse under $a$, and this inverse is $q-2(k-1)-2 = q-2k$.
Thus we must have $p = q-2k$. \qed
\end{proof}

\begin{proof}[Proposition \ref{prop:restricted_product}]

The upper bound of $(m-1)2^{n-1}+1$ was established in~\cite{EHJ16}.
The NFA construction that was
used in~\cite{EHJ16} for the product is shown in Figure~\ref{fig:restricted_product}.
State 0 can be omitted and its transitions replaced by appropriate transitions from the final state of $\cD'_m$ to states 1, 2 and $n-1$.
We will show that $\{0'\}$ and all subsets of the form $\{p'\} \cup S$, where $p' \in Q'_m \setminus \{0'\} = \{1',\dots,(m-1)'\} $ and $S\subseteq Q_n \setminus \{0\} = \{1,\dots,n-1\}$, are reachable and pairwise distinguishable.

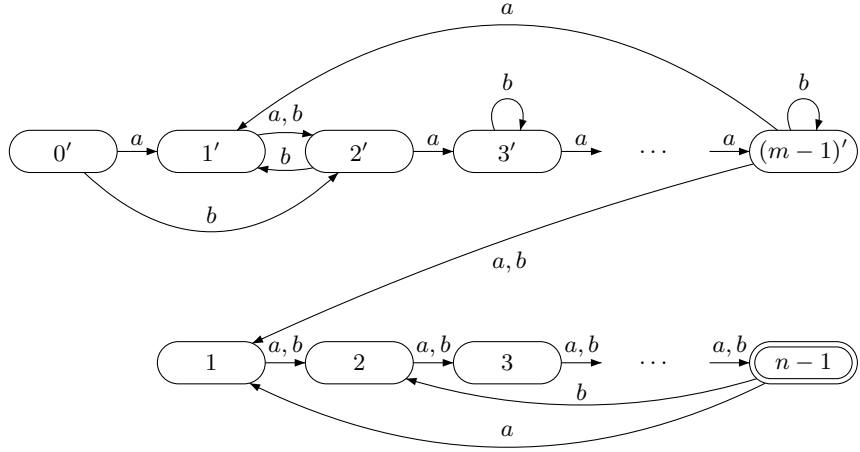
\begin{figure}[th]
\unitlength 8.0pt
\begin{center}\begin{picture}(37,24)(0,1)
\gasset{Nh=2.0,Nw=5.1,Nmr=1.25,ELdist=0.4,loopdiam=1.5}
{\small
\node(0')(1,16){$0'$}
\node(1')(8,16){$1'$}\
\node(2')(15,16){$2'$}
\node[Nframe=n](3dots')(29,16){$\dots$}
\node(3')(22,16){$3'$}	
\node(m-1')(36,16){$(m-1)'$}

\drawedge[curvedepth= 0,ELdist=.3](0',1'){$a$}
\drawedge[curvedepth= -6.,ELdist=-1](m-1',1'){$a$}
\drawedge[curvedepth= .9,ELdist=.3](1',2'){$a,b$}
\drawedge[curvedepth= .9,ELdist=-1.1](2',1'){$b$}
\drawedge[curvedepth= .0,ELdist=.3](2',3'){$a$}
\drawedge[curvedepth= .0,ELdist=.3](3',3dots'){$a$}
\drawedge[curvedepth= 0,ELdist=.3](3dots',m-1'){$a$}
\drawedge[curvedepth= -3.8,ELdist=.4](0',2'){$b$}
\drawloop(3'){$b$}
\drawloop(m-1'){$b$}
\node(1)(8,6){$1$}\
\node(2)(15,6){$2$}
\node[Nframe=n](3dots)(29,6){$\dots$}
\node(3)(22,6){$3$}	
\node(n-1)(36,6){$n-1 $}\rmark(n-1)

\drawedge[curvedepth= 4,ELdist=-1](n-1,1){$a$}
\drawedge[curvedepth= 2,ELdist=-1](n-1,2){$b$}
\drawedge[curvedepth= .0,ELdist=.3](1,2){$a,b$}
\drawedge[curvedepth= .0,ELdist=.3](2,3){$a,b$}
\drawedge[curvedepth= .0,ELdist=.3](3,3dots){$a,b$}
\drawedge[curvedepth= 0,ELdist=.3](3dots,n-1){$a,b$}
\drawedge[curvedepth= -0.8,ELdist=.3](m-1',1){$a,b$}
}
\end{picture}\end{center}
\caption{NFAs for product of $L'_m(a,b)$ and $L_n(a,-,b)$ .}
\label{fig:restricted_product}
\end{figure}

Set $\{0'\}$ is the initial state,  and $\{p' \}$ is reached by $a^p$.
Hence our claim holds for $|S| = 0$.
Next, we prove the claim for sets $\{p'\} \cup S$ with $|S| = 1$.
First consider sets of the form $\{1',q\}$, $q \in Q_n \setminus \{0\}$. 
\bi
\item
First note that $\{1'\} \tr{a^{m-1}} \{1',1\}$.
\item
We have $\{1',q\} \tr{a} \{2',q+1\} \tr{b} \{1',q+2\}$ for $q+2 \le n-1$. By starting from $\{1',1\}$ and repeatedly applying $ab$, we can reach all sets $\{1',q\}$ with $q$ odd.
\item
If $n-1$ is odd, we can reach $\{1',n-1\}$ and then $\{1',n-1\} \tr{a} \{2',1\} \tr{b} \{1',2\}$.
Then by repeatedly applying $ab$ we reach all sets $\{1',q\}$ with $q$ even, and we are done.
\item
If $n-1$ is even, $n-2$ is odd, so we can reach $\{1',n-2\}$. Then $\{1',n-2\} \tr{b} \{2',n-1\} \tr{b} \{1',2\}$. By repeatedly applying $ab$ we reach all sets $\{1',q\}$ with $q$ even.
\ei
Now, note that $a$ restricted to $Q_n \setminus \{0\}$ is a permutation, so each $q \in Q_n \setminus \{0\}$ has an inverse under $a$.
We can reach $\{1',qa^{-(p-1)}\}$ and then by $a^{p-1}$ we reach $\{p',q\}$. 
Thus $\{p',q\}$ is reachable for each $p' \in Q'_m \setminus \{0'\}$ and $q \in Q_n\setminus \{0\}$.

We prove the claim for the remaining sets $\{p'\} \cup S$ by induction, taking the case $|S| = 1$ as our base case.
Suppose $|S| > 1$, and suppose $\{p'\} \cup T$ is reachable whenever $p' \in Q'_m \setminus \{0'\}$ and $|T| < |S|$. 
We refer to this as the \e{primary induction hypothesis}, and the claim as the \e{primary claim}, since we will need to do another inductive proof within this proof.

First we consider the case where $p' = 1'$. Our argument for this case will be split into sub-cases based on the smallest element of $S$.
If $1$ is the smallest element of $S$, then by the primary induction hypothesis, we can reach $\{1'\} \cup T$ for $T = (S \setminus \{1\})a^{-(m-1)}$.
Thus $\{1'\} \cup T \tr{a^{m-1}} \{1',1\} \cup (S \setminus \{1\}) = \{1'\} \cup S$ is reachable.

Suppose $q = 1+2k$ is the smallest element of $S$.
By Lemma \ref{lem:ab}, since $S \subseteq \{1+2k,\dotsc,n-1\}$, each element of $S$ has an unique inverse under $(ab)^k$. 
Thus the size of $S(ab)^{-k}$ has the same size as $S$.
Furthermore, $q-2k = 1$ is the inverse of $q$ under $(ab)^k$, and thus $1 \in S(ab)^{-k}$.
In particular, $1$ is the smallest element of $S(ab)^{-k}$, so we can reach $\{1'\} \cup S(ab)^{-k}$. 
Also, $1'$ is fixed by $ab$.
Thus $\{1'\} \cup S(ab)^{-k} \tr{(ab)^k} \{1'\} \cup S$ is reachable.

This shows $\{1'\} \cup S$ is reachable whenever the smallest element of $S$ is odd.
We deal with the case where the smallest element of $S$ is even in two steps.
First, we reduce to the case where $S$ contains at least one odd element.
Assume that $\{1'\} \cup T$ is reachable whenever $T$ contains at least one odd element and $|T| \le |S|$, and assume $S$ contains no odd elements.
Let $q = 2s$ be the smallest element of $S$.
Observe that $S \setminus \{q\} \subseteq \{1+2s,\dotsc,n-1\}$ and thus by Lemma \ref{lem:ab}, the set $(S \setminus \{q\})(ab)^{-s}$ has the same size as $S \setminus \{q\}$.

If $n-1$ is odd, we can reach $\{1',n-1\} \cup (S \setminus \{q\})(ab)^{-s}$, since $\{n-1\} \cup (S \setminus \{q\})(ab)^{-s}$ has the same size as $S$ and contains an odd element.
Then since $n-1 \tr{ab} 2 \tr{(ab)^{s-1}} 2 + 2(s-1) = 2s = q$, applying $(ab)^s$ gives $\{1'\} \cup S$.

If $n-1$ is even, we can reach $\{1',n-2\} \cup (S \setminus \{q\})(ab)^{-s}$. Since $n-2 \tr{a} n-1 \tr{b} 2$, we have $n-2 \tr{ab} 2 \tr{(ab)^{s-1}} 2s = q$ and thus applying $(ab)^s$ gives $\{1'\} \cup S$.

Finally, we deal with the case where $S$ contains at least one odd element.
Let $p$ be the smallest odd element of $S$.
Let $T = \{ q \in S : p \le q \}$, the set of elements above (and including) the smallest odd element.
Thus $S \setminus T$ is the set of elements of $S$ that are strictly below the smallest odd element (in particular, each element of $S \setminus T$ is even).
We prove that $\{1'\} \cup S$ is reachable by induction on the size of $S \setminus T$.
This is the \e{secondary claim}.

The base case is $|S \setminus T| = 0$. Here every element of $S$ is greater than or equal to $p$.
Thus we are in the case where the smallest element of $S$ is odd, which we have solved.

Suppose that whenever $|S \setminus T| < k$, we can reach $\{1'\} \cup S$.
That is, whenever there are fewer than $k$ elements strictly below the smallest odd element of $S$, we can reach $\{1'\} \cup S$.
This is the \e{secondary induction hypothesis}.
We want to prove the secondary claim holds when $|S \setminus T| = k$ for $k \ge 1$.

Let $r$ be the smallest element of $S$. Since $k \ge 1$, we have $r \in S \setminus T$ and thus $r$ is even. Suppose $r = 2s$. 
Then $S \setminus \{r\} \subseteq \{1+2s,\dotsc,n-1\}$, so by Lemma \ref{lem:ab}, every element of $S \setminus \{r\}$ has a unique inverse under $(ab)^s$.

In $S$, the number of elements strictly below the smallest odd element is $|S \setminus T| = k$. 
Say a set $U$ is \e{nice} if the number of elements strictly below the smallest odd element of $U$ is \e{strictly less than} $k$.
Consider the set $S \setminus \{r\}$. Since $r$ is strictly below the smallest odd element of $S$, it follows that $S \setminus \{r\}$ is nice. 
If we apply $(ab)^{-s}$ to $S \setminus \{r\}$, this simply subtracts $2s$ from every element of $S \setminus \{r\}$, and so preserves parity. 
Thus $(S \setminus \{r\})(ab)^{-s}$ is also nice.
Finally, if we take the union $(S \setminus \{r\})(ab)^{-s} \cup \{n-1\}$, this set is nice, since $n-1$ is the largest element of this set and so cannot possibly be strictly below anything.

The secondary induction hypothesis says precisely that $\{1'\} \cup U$ is reachable whenever $U$ is nice.
Thus $\{1',n-1\} \cup (S \setminus \{r\})(ab)^{-s}$ is reachable.
Notice that $n-1 \tr{(ab)^s} 2s = r$.
Hence by $(ab)^s$ we reach $\{1'\} \cup S$.

This proves the \e{secondary claim}.
We have thus shown that if each state $\{p'\} \cup T$ with $|T| < |S|$ is reachable, then each state $\{1'\} \cup S$ is reachable.
To see that each $\{p'\} \cup S$ is reachable, reach $\{1'\} \cup Sa^{-(p-1)}$ and apply $a^{p-1}$ to reach $\{p'\} \cup S$.
This completes the proof of the \e{primary claim}, and shows that $(m-1)2^{n-1}+1$ states are reachable.

Now, we show that all reachable states are distinguishable.
Let $\{p'\} \cup S$ and $\{q'\} \cup T$ be two distinct states.
Suppose $S \ne T$ and suppose $r \in S \oplus T$, the symmetric difference of $S$ and $T$.
Then $a^{n-1-r}$ distinguishes the states.

Now suppose $S = T$ and $p' \ne q'$. Assume without loss of generality that $p < q$.
Define $U = Sa^{m-1-q}$. Then we have
\[ \{q'\} \cup S \tr{a^{m-1-q}} \{(m-1)'\} \cup U \tr{b} \{(m-1)',1\} \cup Ub. \]
Observe that $(Q_n \setminus\{0\})b = \{2,\dotsc,n-1\}$.
Thus $Ub \subseteq \{2,\dotsc,n-1\}$ and in particular, $Ub \cap \{1\} = \emp$.

On the other hand, let $r' = (p+m-1-q)'$.
If we apply $a^{m-1-q}b$ to $\{p'\} \cup S$, we reach
\[ \{p'\} \cup S \tr{a^{m-1-q}} \{r'\} \cup U \tr{b} \{r'\}b \cup Ub. \]
Note that since $p < q$, we have $r = p+m-1-q < m-1$. Hence $b$ either fixes $r'$, or if $r' \in \{1',2'\}$, it swaps $r'$ with the other element of $\{1',2'\}$. In either case $r'b < m-1$.
It follows that $\{r'\}b = \{r'b\}$.

Thus $\{p'\} \cup S$ is sent to $\{r'b\} \cup Ub$ by $a^{m-1-q}b$, and $\{q'\} \cup S$ is sent to $\{(m-1)'\} \cup (\{1\} \cup Ub)$, with $Ub \cap \{1\} = \emp$.
These two states have different subsets of $Q_n$, so they are distinguishable.
This proves that all states are distinguishable, and hence the bound is met. \qed
\end{proof}

\begin{proposition}[Unrestricted Product]
\label{prop:unrestricted_product}
For $m,n \ge 4$, let $L'_m$ (respectively, $L_n$) be a non-returning language of complexity $m$ (respectively, $n$)  over an alphabet $\Sig'$, (respectively, $\Sig$). 
Then the complexity of product is at most $m2^{n+1}+1$, and this bound is met by
 $L'_m(a,b)$ and $L_n(a,-,b,d)$.
\end{proposition}
Here $L_n(a,-,b,d)$ is the language of DFA $\cD_n(a,-,b,d)$ with transformations $a \co (1,\dotsc,n-1)(0 \to 1)$, $b \co (2,\dotsc,n-1)(1 \ra 2)(0 \ra 1)$ and $d \co (0 \ra 2)$. That is, everything is the same as DFA $\cD_n(a,b,c,d)$, except $c$ is removed and $b$ now induces the transformation that $c$ originally induced.

\begin{proof}
First we derive an upper bound; the derivation is similar to that for regular languages~\cite{Brz16}.
Let $\cD'_m=( Q'_m, \Sig', \delta', 0',F')$  and $\cD_n=(Q_n,\Sig,\delta,0,F) $ be minimal DFAs of non-returning languages $L'_m\subseteq (\Sig')^*$  and $L_n\subseteq \Sig^*$, respectively.
We use the same construction as in the restricted case. 
As in that case, we may possibly  reach the initial state $\{0'\}$ and $(m-1)2^{n-1}$ states using letters in $\Sig'\cap\Sig$.
But starting in any set of the form $\{p'\} \cup S$, where $S\subseteq \{1,\dots,n-1\}$,
and applying a letter from $\Sig\setminus \Sig'$ we may also reach $S$. 
Hence at most 
$(m-1)2^{n-1}+2^{n-1}+1= m2^{n-1}+1$ states can be reached.

Languages $L'_m(a,b)$ and $L_n(a,-,b,d)$ meet this bound. We can reach 
$(m-1)2^{n-1}+1$ states by words in $\{a,b\}^*$. From $\{p'\} \cup S$, we can  use $d$ as described above to reach $S$ and thus reach $2^{n-1}$ more states. It remains to show that all these states are distinguishable.

We will say a state is a $Q'_m$-state if it contains an element of $Q'_m$, and a $Q_n$-state otherwise.
The reachable $Q'_m$-states have the form $\{p'\} \cup S$ with $p' \in Q'_m \setminus \{0'\}$ and $S \subseteq Q_n$.
The reachable $Q_n$-states have the form $S \subseteq Q_n$.

All $Q'_m$-states are all pairwise distinguishable using words over $\{a,b\}^*$, using the same arguments as the restricted case.
Given two distinct $Q_n$-states $S$ and $T$, we can take $r \in S \oplus T$ and distinguish the states using $a^{n-1-r}$, as in the restricted case.

Now let $\{p'\} \cup S$ be a $Q'_m$-state and let $T$ be a $Q_n$-state.
If $S \ne T$, we can take $r \in S \oplus T$ and use $a^{n-1-r}$ to distinguish.
If $S = T$, let $U = Sa^{m-1-p}$ so that 
$\{p'\} \cup S \tr{a^{m-1-p}} \{(m-1)'\} \cup U \tr{b} \{(m-1)',1\} \cup Ub$.
As in the restricted case, we have $Ub \subseteq \{2,\dotsc,n-1\}$.
So $a^{m-1-p}b$ sends $\{p'\} \cup S$ to $\{(m-1)'\} \cup (\{1\} \cup Ub)$, with $Ub \cap \{1\} = \emp$, but it sends $S$ simply to $Ub$.
These two states have different subsets of $Q_n$, so they are distinguishable.

So all $Q'_m$ states are pairwise distinguishable, all $Q_n$-states are pairwise distinguishable, and all $Q'_m$-states are distinguishable from all $Q_n$-states; hence all states are distinguishable. \qed
\end{proof}

\section{Boolean Operations}

\begin{proposition}[Restricted Boolean Operations]   
\label{prop:restricted_boolean}
Let $\cD_n(a,b)$ be the DFA of Definition~\ref{def:most_complex} and let $L_n(a,b)$ be its language. Then for $m,n \ge 4$ and for any proper binary boolean operation $\circ$ the complexity of 
$L'_m(a,b)\circ L_n(b,a)$ is $mn-(m+n-2)$.
If $m \neq n$ then $\kappa(L'_m(a,b)\circ L_n(a,b)) = mn-(m+n-2)$.
\end{proposition}
\begin{proof}
The upper bound was established in~\cite{EHJ16}.
For the lower bound,
Figure~\ref{fig:restricted boolean} shows the two argument DFAs.
As usual we construct their direct product.
State $(0',0)$ is initial and can never be reached again.
If we apply $a$,  we reach state $(1',2)$, and the states reachable from this state form the direct product  of DFA
$\mathcal{E}'_{m-1}(a,b)=(\{1,\dots,(m-1)'\},\{a,b\},\delta', 1', \{(m-1)'\})$ and DFA
$\mathcal{E}_{n-1}(b,a)=(\{1,\dots,n-1\},\{a,b\},\delta, 2, \{n-1\})$, where $\delta'$ and $\delta$ are $\delta_{m'}$ and $\delta_n$ restricted to $Q'_m\setminus\{0'\}$ and 
$Q_n\setminus \{0\}$.
Since the transition semigroups of $\cE'_m$ and $\cE_n$ are the symmetric groups $S_m$ and $S_n$, respectively, 
the result from~\cite[Theorem 1]{BBMR14} applies, except in the cases where $(m,n)$ is in
$\{(4,5), (5,4), (5,5)\}$, which have been verified by computation.
Our first  claim follows for the remaining cases by~\cite[Theorem 1]{BBMR14}. 
If $m\neq n$, \cite[Theorem 1]{BBMR14} applies to $\cD'_m(a,b)$ and $\cD_n(a,b)$, and the second claim follows.
In both cases the direct product of $\cE'_m$ and $\cE_n$ has $(m-1)(n-1)$ states; hence in the direct product of $\cD'_m$ and $\cD_n$ there are 
$(m-1)(n-1)+1= mn-(m+n-2)$ states.
By~\cite[Theorem 1]{BBMR14} all these states are reachable and pairwise distinguishable for every proper operation $\circ$.

Finally, note also that $(0',0)$ is distinguishable from all other states. Since $(0',0)a = (1',2)$ and $(1',2)a^{-1} = \{(m-1)',1)\}$, we see that $(0',0)$ is distinguishable from $(p',q) \ne ((m-1)',1)$ by first applying $a$, then applying a word that distinguishes $(0',0)a = (1',2)$ from $(p',q)a$. It is distinguishable from $((m-1)',1)$ by first applying $b$, then applying a word that distinguishes $(0',0)b = (2',1)$ from $((m-1)',1)b = ((m-1)',2)$.
\qed
\end{proof}

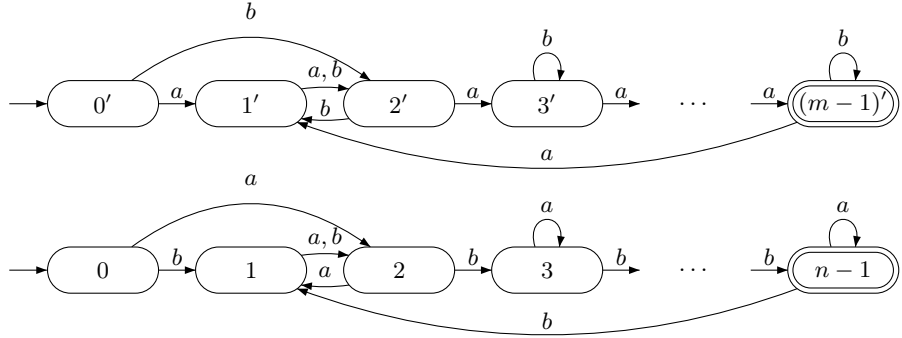
\begin{figure}[th]
\unitlength 7.0pt
\begin{center}\begin{picture}(37,16)(1.6,5)
\gasset{Nh=2.5,Nw=6,Nmr=1.25,ELdist=0.4,loopdiam=1.5}
{\small
\node(0')(1,16){$0'$}\imark(0')
\node(1')(9,16){$1'$}\
\node(2')(17,16){$2'$}
\node[Nframe=n](3dots)(33,16){$\dots$}
\node(3')(25,16){$3'$}	
\node(m-1)(41,16){$(m-1)'$}\rmark(m-1)
\drawedge[curvedepth= 0,ELdist=.3](0',1'){$a$}
\drawedge[curvedepth= 3.5,ELdist=-1](m-1,1'){$a$}
\drawedge[curvedepth= .9,ELdist=.3](1',2'){$a,b$}
\drawedge[curvedepth= .9,ELdist=-1.1](2',1'){$b$}
\drawedge[curvedepth= .0,ELdist=.3](2',3'){$a$}
\drawedge[curvedepth= .0,ELdist=.3](3',3dots){$a$}
\drawedge[curvedepth= 0,ELdist=.3](3dots,m-1){$a$}
\drawedge[curvedepth= 3.6,ELdist=1.0](0',2'){$b$}
\drawloop(3'){$b$}
\drawloop(m-1){$b$}
\node(0)(1,7){0}\imark(0)
\node(1)(9,7){1}\
\node(2)(17,7){2}
\node[Nframe=n](3dots)(33,7){$\dots$}
\node(3)(25,7){$3$}
\node(n-1)(41,7){$n-1$}\rmark(n-1)
\drawedge[curvedepth= 0,ELdist=.3](0,1){$b$}
\drawedge[curvedepth= 3.5,ELdist=-1.2](n-1,1){$b$}
\drawedge[curvedepth= .9,ELdist=.3](1,2){$a,b$}
\drawedge[curvedepth= .9,ELdist=-1.1](2,1){$a$}
\drawedge[curvedepth= .0,ELdist=.3](2,3){$b$}
\drawedge[curvedepth= .0,ELdist=.3](3,3dots){$b$}
\drawedge[curvedepth= 0,ELdist=.3](3dots,n-1){$b$}
\drawedge[curvedepth= 3.6,ELdist=1.0](0,2){$a$}
}
\drawloop(3){$a$}
\drawloop(n-1){$a$}

\end{picture}\end{center}
\caption{DFAs $\cD'_m(a,b)$ and $\cD_n(b,a)$ for boolean operations.}
\label{fig:restricted boolean}
\end{figure}

\begin{proposition}[(Unrestricted Boolean Operations)]
\label{prop:unrestricted_boolean}
For $m,n \ge 4$, let $L'_m(\Sig')$ (respectively, $L_n(\Sig)$) be a non-returning language of complexity $m$ (respectively, $n$)  over an alphabet $\Sig'$, (respectively, $\Sig$). 
Then the complexity of union and symmetric difference is $mn+1$ 
and this bound is met by $L'_m(a,b,c)$ and $L_n(b,a,d)$; 
the complexity of difference is $mn-n+1$, and this bound is met by $L'_m(a,b,c)$ and
$L_n(b,a)$; the complexity of intersection is $mn-(m+n-2)$ and this bound is met by  
$L'_m(a,b)$ and $L_n(b,a)$.
\end{proposition}
\begin{proof}
First we derive upper bounds for four boolean operations.
Let $\cD'_m(\Sig')= ( Q'_m, \Sig', \delta'_m, 0',F')$  and 
$\cD_n(\Sig) = (Q_n, \Sig,\delta_n, 0, F)$ be minimal DFAs for non-returning languages $L'_m(\Sig')$ and $L_n(\Sig)$  of complexity $m$ and $n$, respectively.
Assume that $\Sig'\setminus \Sig $ and 
$\Sig \setminus \Sig' $ are non-empty; this assumption leads to the largest upper bound.
We  add an empty state $\emp'$ to $\cD'_m$ to send all transitions under the letters from 
$\Sig \setminus \Sig' $ to that state; thus we get an $(m+1)$-state DFA  $\cD'_{m,\emp'}$. Similarly, we add an empty state $\emp$  to $\cD_n$ to get $\cD_{n,\emp}$.
Now we construct the direct product of $\cD'_{m,\emp'}$ and $\cD_{n,\emp}$. 
There are at most  $(m+1)(n+1)$ states in this direct product. 
Note, however, that the $m$ states of the form $(p',0)$, $p\neq 0$, and the $n$ states of the form $(0',q)$, $q\neq 0$ are not reachable. Thus we have an upper bound of $(m+1)(n+1)-(m+n)= mn+1$. 
We will show that this bound can be reached for union and symmetric difference.
However, the alphabet of the difference $L'_{m,\emp'} \setminus L_{n,\emp}$ is
$\Sig'$ and states of the form $(\emp',q)$,  $q \in \{1,\dots,n-1\} \cup \{\emp\}$, which are reachable only by letters in $\Sig \setminus \Sig'$, are not reachable in the difference DFA. Hence the bound reduces to $mn -n+1$. 
Similarly, the alphabet of the intersection $L'_{m,\emp'} \cap L_{n,\emp}$ is
$\Sig'\cap \Sig$ and states of the form $(p',\emp)$, $p \in \{1',\dots, (m-1)'\} \cup \{\emp'\}$, $(\emp',q)$, $q \in \{1,\dots,n-1\} \cup \{\emp\}$ are not reachable. Hence the bound reduces to $mn -(m+n-2)$, which is the same as the bound in the restricted case. 

\begin{figure}[th]
\unitlength 8.0pt
\begin{center}\begin{picture}(37,19.5)(0,4)
\gasset{Nh=2.0,Nw=5.1,Nmr=1.25,ELdist=0.4,loopdiam=1.5}
{\small
\node(0')(1,16){$0'$}
\node(1')(8,16){$1'$}\
\node(2')(15,16){$2'$}
\node[Nframe=n](3dots)(29,16){$\dots$}
\node(3')(22,16){$3'$}	
\node(m-1)(36,16){$(m-1)'$}\rmark(m-1)

\drawedge[curvedepth= 0,ELdist=.3](0',1'){$a,c$}
\drawedge[curvedepth= 4,ELdist=-1](m-1,1'){$a$}
\drawedge[curvedepth= 2,ELdist=-1](m-1,2'){$c$}
\drawedge[curvedepth= .9,ELdist=.3](1',2'){$a,b,c$}
\drawedge[curvedepth= .9,ELdist=-1.1](2',1'){$b$}
\drawedge[curvedepth= .0,ELdist=.3](2',3'){$a,c$}
\drawedge[curvedepth= .0,ELdist=.3](3',3dots){$a,c$}
\drawedge[curvedepth= 0,ELdist=.3](3dots,m-1){$a,c$}
\drawedge[curvedepth= 5,ELdist=.4](0',2'){$b$}
\drawloop(3'){$b$}
\drawloop(m-1){$b$}
\node(0)(1,7){0}\imark(0)
\node(1)(8,7){1}\
\node(2)(15,7){2}
\node[Nframe=n](3dots)(29,7){$\dots$}
\node(3)(22,7){$3$}
\node(n-1)(36,7){$n-1$}\rmark(n-1)
\drawedge[curvedepth= 0,ELdist=.3](0,1){$b$}
\drawedge[curvedepth= 3.5,ELdist=-1.2](n-1,1){$b$}
\drawedge[curvedepth= .9,ELdist=.3](1,2){$a,b$}
\drawedge[curvedepth= .9,ELdist=-1.1](2,1){$a$}
\drawedge[curvedepth= .0,ELdist=.3](2,3){$b$}
\drawedge[curvedepth= .0,ELdist=.3](3,3dots){$b$}
\drawedge[curvedepth= 0,ELdist=.3](3dots,n-1){$b$}
\drawedge[curvedepth= 4,ELdist=.6](0,2){$a,d$}
}
\drawloop(1){$d$}
\drawloop(2){$d$}
\drawloop(3){$a,d$}
\drawloop(n-1){$a,d$}

\end{picture}\end{center}
\caption{DFAs $\cD'_m(a,b,c)$ and $\cD_n(b,a,d)$ for unrestricted boolean operations.}
\label{fig:unrestricted_boolean}
\end{figure}
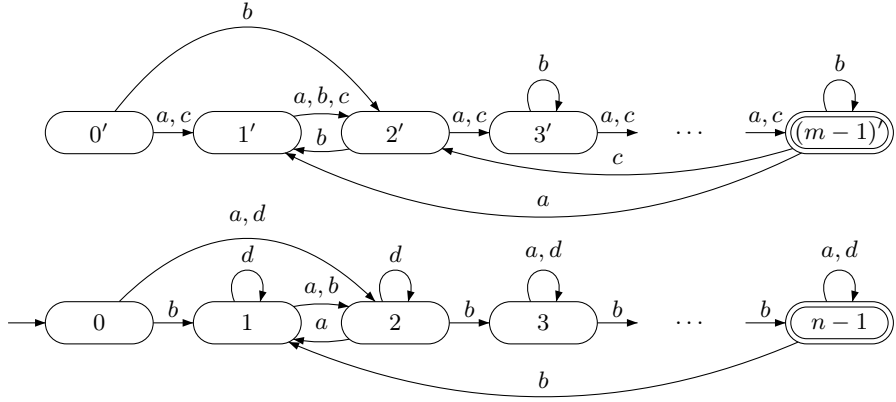

We now claim that our witnesses of Figure~\ref{fig:unrestricted_boolean} 
meet this bound. 

\noin
{\bf Union} Consider the direct product of $\cD'_m$ and $\cD_n$.
State $(0',0)$ is initial and reaches $(1',2)$ by $a$.
By~\cite[Theorem 1]{BBMR14}  and computation in the cases where $(m,n)$ is in $\{(4,5), (5,4), (5,5)\}$ all $(m-1)(n-1)$ states of the form $(p',q)$, 
$p' \not\in \{0', \emp'\}$, and $q\not \in \{\emp,0\}$ are reachable and pairwise distinguishable.
The states that have an empty component are reachable as follows:
$(\emp',1)$ can reached from any reachable state $(q',1)$ by $d$, and 
$(\emp',1) \stackrel{ b^{q-1} }{ \longrightarrow } (\emp',q)$, for $q = 1, \dots, n-1$.
The state $(1',\emp)$ is reached from $(0',0)$ by $c$, and 
$(1',\emp) \stackrel{ a^{p-1} }{ \longrightarrow } (p',\emp)$, for $p = 1, \dots, m-1$.
Finally, $(\emp',\emp)$ is reached from any state $(\emp', q)$ by $c$.
Thus we have a total of $1+(m-1)(n-1) + (m-1) +(n-1) +1= mn+1$ states.

Let $S=\{(p',q) \mid 1\le p\le m-1, 1\le q \le n-1\}$, 
$V=\{(p',\emp) \mid 1\le p\le m-1\}$, and 
$H = \{(\emp', q) \mid 1\le q \le n-1\}$.
States in $H$ are distinguishable by words in $b^*$, and those in
$V$, by words in $a^*$.
State $(\emp',\emp)$ is the only empty state; hence it is distinguishable from every other reachable state. 
State $(0',0)$ is distinguished from every other reachable non-final state because only it accepts $a^{m-1}$.
Non-final states in $H$ are distinguished from those in $S$ and $V$ by words in $a^*$.
Non-final states in $V$ are distinguished from those in $S$ by words in $b^*$. 
State $(\emp',n-1)$ is distinguished from $((m-1)',\emp)$ by $d$, and from final states in $S$ by words in $ca^*$.
State $((m-1)',\emp)$ is distinguished  from final states in $S$ by words in $db^*$.

\noin
{\bf Symmetric Difference} 
Now state $((m-1)',n-1)$ is no longer final, but the rest of the argument is the same as for union. 

\noin
{\bf Difference} Here states $( p', n-1)$, $p'\in \{1',\dots,(m-1)'\} \cup  \{\emp'\}$ are no longer final. Since the alphabet of $L'_m\setminus L_n$ is $\Sig'$, states in $H$ are  not reachable, and there are only $mn-n+1$ reachable states.
The states in $S$ are pairwise distinguishable by~\cite[Theorem 1]{BBMR14} and those in $V$ by the argument used for union.
The arguments used for union also apply here to distinguish states in $S$ from those in $V$.

\noin
{\bf Intersection} Now the alphabet of $L'_m\cap L_n$ is $\Sig'\cap \Sig$, and  neither the states in $H$ nor those in $V$ are reachable. The claim now follows by~\cite[Theorem 1]{BBMR14}.
\qed
\end{proof}

The complexity of any other binary boolean operation can be determined from
the complexities of union, intersection, difference and symmetric difference; 
however, the complexity of $L_m\circ L'_n$ may differ by 1 from 
the complexity of $\ol{L_m \circ L'_n}$.
For more details see~\cite{BrSi16}.

\section{Conclusions}

We have shown that there exists a most complex non-returning language stream
$(L_4(\Sig), \dots, L_n(\Sig),\dots)$. 
The cardinality of the syntactic semigroup of $L_n(\Sig)$ is $(n-1)^n$ and its atoms are have the highest state complexity possible for non-returning languages; both of these bounds can be reached only if $\Sig$ has at least $\binom{n}{2}$ letters.
The bounds for the common restricted operations, however, can be met by streams over $\{a,b,c\}$ or $\{a,b\}$: 
$\kappa(L_m(a,b) \circ L_n(b,a) = mn - (m+n-2)$ for all proper boolean operations $\circ$; 
$\kappa(L_n(a,b))^*= 2^{n-1}$;
$\kappa(L_n(a,b,c)^R)=2^n$;
and $\kappa(L'_m(a,b)L_n(a,-,b)) = (m-1)2^{n-1} + 1$.
The bounds for unrestricted boolean operations can be met by $L'_m(a,b,c)$ and $L_n(b,a,d)$, whereas those for the unrestricted product, by $L'_m(a,b)$ and
$L_n(a,-,b,d)$.
\smallskip

\noin
{\bf  Acknowledgment}
We are very grateful to Corwin Sinnamon for careful proofreading and constructive comments.

\providecommand{\noopsort}[1]{}


\end{document}